\documentclass[conference]{IEEEtran}
\IEEEoverridecommandlockouts
\usepackage{cite}
\usepackage{amsmath,amssymb,amsfonts}
\usepackage{algorithmic}
\usepackage{graphicx}
\usepackage{textcomp}
\usepackage{xcolor}
\usepackage{hyperref}
\usepackage{indentfirst} 
\usepackage{algorithm}
\usepackage{array}
\usepackage{float}
\usepackage{bm}  
\usepackage{mathrsfs}
\usepackage{amsthm,amssymb}
\usepackage{stfloats}
\usepackage{url}
\usepackage{verbatim}
\usepackage{subfigure}
\usepackage{xfrac}  
\usepackage{titlesec}
\usepackage{cleveref}
\newcommand{\T}{{\sf T}}

\DeclareMathOperator{\tr}{tr}
\DeclareMathOperator{\diag}{diag}
\newtheorem{theorem}{Theorem}
\newtheorem{proposition}{Proposition}

\newtheorem*{remark}{Remark}
\newtheorem{assumption}{Assumption}

\def\BibTeX{{\rm B\kern-.05em{\sc i\kern-.025em b}\kern-.08em
    T\kern-.1667em\lower.7ex\hbox{E}\kern-.125emX}}
    
\begin{document}
\title{Asymptotic CRB Analysis of Random RIS-Assisted Large-Scale Localization Systems}

\author{\IEEEauthorblockN{Zhengyu Wang, Hongzheng Liu, Rujing Xiong, Fuhai Wang and Robert Caiming Qiu }\\
\IEEEauthorblockA{\textit{ School of Electronic Information and Communications } \\
\textit{ Huazhong University of Science and Technology, Wuhan 430074, China }\\
Email: \{wangzhengyu,hongzhengliu,Rujing,wangfuhai,caiming\}@hust.edu.cn }

\thanks{This work was supported by the National Natural Science Foundation of
China under Grant 12141107.}
}

\maketitle

\begin{abstract}

This paper studies the performance of a randomly RIS-assisted multi-target localization system, in which the configurations of the RIS are randomly set to avoid high-complexity optimization.
We first focus on the scenario where the number of RIS elements is significantly large, and then obtain the scaling law of  Cramér–Rao bound (CRB) under certain conditions, which shows that CRB decreases in the third or fourth order as the RIS dimension increases.
Second, we extend our analysis to large systems where both the number of targets and sensors is substantial. 
Under this setting, we explore two common RIS models: the constant module model and the discrete amplitude model, and illustrate how the random RIS configuration impacts the value of CRB.
Numerical results demonstrate that asymptotic formulas provide a good approximation to the exact CRB in the proposed randomly configured RIS systems.

\end{abstract}

\begin{IEEEkeywords}
Reconfigurable intelligent surface,\,random configuration,\,CRB,\,scaling law,\,random matrix theory
\end{IEEEkeywords}

\section{Introduction}
Reconfigurable intelligent surface (RIS) is a novel technology in wireless communications and has attracted growing research interest in recent years.
A RIS is a plane containing numerous regularly arranged electromagnetic elements and can help manipulate impinging electromagnetic waves by designing the structure and parameters of the elements.

Since the RIS can provide a virtual Line-of-Sight (LOS) link and improve the positioning performance by suppressing interference and enhancing useful signals, extensive research on RIS-assisted localization has been conducted \cite{zhang2021metalocalization}\cite{keykhosravi2023leveraging}.
The CRB sets a lower bound on the variance of any unbiased estimator and is often used as a performance metric in the RIS optimization strategy.
For example, the author in \cite{10138058} jointly optimizes the transmit beamforming at the Access Point (AP) and the reflective beamforming at the RIS to minimize the CRB on the estimation error.

The spacing between independent elements decreases as the operating frequency increases. 
In the millimeter-wave or even terahertz frequency bands, there will be hundreds or thousands of units on an RIS.
In such scenarios, solving the optimization problem of RIS becomes difficult or even impossible, due to the exponential increase of the search domain.
As an alternative, recent proposals suggest the use of random phase configurations to avoid channel estimation and non-convex optimization \cite{chen2022non}\cite{psomas2020random}\cite{zappone2021intelligent}\cite{tegos2021distribution}.
The performance of a wireless communication network assisted by randomly reconfigurable surfaces is first studied in \cite{tegos2021distribution}.
Specifically, closed-form expressions for the outage probability, the
average received signal-to-noise ratio (SNR) and the ergodic capacity are derived to evaluate the performance of the considered system.
To the best of the authors’ knowledge, the performance of a randomly RIS-assisted localization system has not yet been investigated in the existing literature.
Motivated by this challenge, in this paper, we consider a randomly large RIS-assisted direction-of-arrival (DoA) estimation.
This paper aims to address the following two questions:
What impact does a random phase setting have on the performance of the localization system?
What is the asymptotic performance of the CRB in a high-dimensional system?

Some work has explored the impact of RIS parameters on CRB. 
For example, in \cite{hu2018beyond}, the author derived the asymptotic expression of the CRB under the large intelligent surfaces (LIS) system, demonstrating that the CRLB for all three Cartesian dimensions decreases quadratically in the surface area of the LIS under mild conditions.
Additionally, the author in \cite{fang2023multiirsenabled} derives the closed-form CRB in a multi-RIS enabled Integrated Sensing and Communications (ISAC) system.
The CRB for angle-of-arrival (AoA) estimation is inversely proportional to the cubic dimension of each RIS in the case of a point target.
However, it is crucial to note that this conclusion is limited to single-user estimation scenarios.
\cite{stoica1989music} and \cite{stoica1990music} derive the asymptotic CRB formula when the number of users or sensors tends to infinity, and simulations verify that the asymptotic CRLB provide accurate approximations to the exact CRB.
Nevertheless, the above asymptotic conclusion cannot be directly applied to the RIS system.
The scaling law of CRB of the RIS-assisted localization system has not been investigated in the literature yet.
The contributions of this work are listed as follows: 
\begin{itemize}
\item 
Firstly, we introduce random RIS and derive the scaling law of the CRB under the assumption given in \Cref{ass:iid}.
The resulting asymptotic CRB expression provides an easily evaluated form, explicitly revealing the influence of RIS dimension, number of targets, and number of sensors on CRB.
\item 
Secondly, for randomly RIS-assisted large systems, we introduce an analysis framework based on random matrix theory.
Additionally, we analyze how the distribution of random RIS, for certain models such as the random phase model, discrete amplitude model, etc., influences the CRB.
\item 
The scaling law of CRB shows that CRB decreases in the third or fourth order as the RIS dimension increases. 
Notably, CRB has nothing to do with the specific values for the incident and departure angles but relies on parameters such as SNR and sensor number.
\end{itemize}

\textit{Notations}: 
Upper-case and lower-case boldface letters denote matrices and column vectors, respectively.
The operators $(\cdot)^T$ and $(\cdot) ^H$ denote the transpose, and Hermitian transpose respectively.
Additionally, $\odot$ represents the Hadamard matrix product.
$\tr(\cdot)$ is the trace of a matrix. 
$j$ denotes the imaginary unit.
$ \diag(\textbf{x}) $ returns a diagonal matrix with the elements in $\textbf{x}$ as its main diagonal entries.

\section{Problem Formulation}
Consider a uniform linear array (ULA) RIS-assisted DoA estimation system, where there are $R$ sensors to sense $K$ single-antenna targets
We assume that $K$ narrowband far-field targets impinge on the RIS from unknown directions of $[\theta_1, \dots, \theta_K]$.
DoA estimation can be formulated as a parameter estimation problem with the following model:
\begin{align}
    \textbf{y}(t) = \textbf{A}_{\Phi} \bm{\Omega}(t) \textbf{A}_{\Theta} \textbf{x}(t) + \textbf{n}(t),
\end{align}
where $\textbf{A}_{\Theta} = [\textbf{a}(\theta_1),\dots,\textbf{a}(\theta_K)]  \in \mathcal{C}^{N \times K} $ and $\textbf{A}_{\Phi}=[\textbf{a}^\T(\phi_1),\dots,\textbf{a}^\T(\phi_R)]^\T  \in \mathcal{C}^{R \times N}$ denote the array manifold matrix from targets to the RIS and from RIS to sensors, respectively. 
$\textbf{n}(t) \in \mathcal{CN}(0,\sigma^2 \textbf{\textit{I}}_R)$ is the addictive white Gaussian noise vector where $\textbf{\textit{I}}_R$ is a $R \times R$ identity matrix and $\sigma^2$ denotes the power of the noise.
The steering vector of target $k$ at DoA $\theta_k$ is given by:
\begin{align*}
    \textbf{a}(\theta_k) = [1,~e^{j\theta_k},\ldots, e^{j(N-1)\theta_k}]^\T.
\end{align*}

The response matrix $\bm{\Omega}=\diag(\bm{\omega})\in \mathcal{C}^{N \times N}$, where $\bm{\omega}=(\beta_1 e^{j \alpha_1}, \beta_2 e^{j \alpha_2}, \dots, \beta_N e^{j \alpha_{N}})^\T \in \mathcal{C}^{N \times 1}$ represents the reflection coefficient,
$\alpha_{n} \in (0,2\pi)$ denoting the phase shift and $\beta_n$ denoting the amplitude response of $n$-th reflecting element of the RIS.

Collect all the unknown parameters into a vector $\bm{\xi}=[\theta_1, \dots, \theta_K]^\T$, the probability density function of the received signal of $R$ sensors can be expressed as: 
\begin{align*}
    f(\textbf{y};\bm{\xi}) = \frac{1}{\pi^{RT} \det(\bm{C})^\T} e^{-\sum_{t=1}^T[\textbf{y}(t)-\bm{\mu}(t)]^H \bm{C}^{-1} [\textbf{y}(t)-\bm{\mu}(t)]},
\end{align*}
where the mean and covariance matrix are respectively
\begin{align*}
    \bm{\mu}(t) &= \textbf{A}_{\Phi} \bm{\Omega}(t) \textbf{A}_{\Theta} \textbf{x}(t),  \\
    \bm{C} &= \sigma^2 \textit{\textbf{I}}_R .
\end{align*} 


To obtain the CRLB we first derive the Fisher information matrix (FIM) $\textbf{F}$ by calculating the partial derivative of the probability density function with respect to $\bm{\xi}$. $\textbf{F}$ has the following expression:
\begin{align} \label{eq:Fisher eq}
    \textbf{F} &= \mathbb{E} \Big\{ \frac{\partial \ln f^H(\textbf{y};\bm{\xi})}{\partial \bm{\xi}} \frac{\partial \ln f(\textbf{y};\bm{\xi})}{\partial \bm{\xi}} \Big\}  \notag \\    
    &= \frac{2}{\sigma^2}\sum_{t=1}^T \Re \Big\{ \textbf{X}^H(t) \textbf{D}^H \textbf{D} \textbf{X}(t) \Big\} \notag \\
    & = \frac{2T }{\sigma^2} \Re \Big\{ (\textbf{D}^H \textbf{D}) \odot \textbf{R}^T \Big\} ,
\end{align}
where $\textbf{X}(t) = \diag(\textbf{x}(t)) $ and $\textbf{D} = \textbf{A}_{\Phi} \bm{\Omega} \dot{\textbf{A}}_{\Theta}$.
$ \dot{\textbf{A}_{\Theta}} = [\dot{\textbf{a}}(\theta_1),\dot{\textbf{a}}_a(\theta_2),\dots,\dot{\textbf{a}}_a(\theta_K)]$, 
where $\dot{\textbf{a}}(\theta_k) = \frac{\partial \textbf{a}(\theta_k)}{\partial \theta_k}$ denotes the derivative of steering vector $\textbf{a}(\theta_k)$ with respect to $\theta_k$.
The last equation is derived under the assumption that the incident signal has finite power. 
Here, $\textbf{R}$ represents the signal covariance matrix, with its $(m,n)$-th element defined as $\textbf{R}_{mn}=\frac{1}{T} \sum_{t=0}^{T} x_n^H(t)x_m(t) $.
The CRB of $k$-th parameter is given as $\text{CRB}_{k,k}=[\textbf{F}^{-1}]_{k,k}$.

\section{CRLB: Large System Analysis}
\subsection{Large RIS}
In this section, we consider the large RIS with the number of reflecting elements $N$ tending to infinity, under the following assumptions.
\begin{assumption}\label{ass:iid}
The reflection coefficient of each element of the RIS is independently and identically distributed, following a random distribution with finite-order moments.
\end{assumption}

In practice, optimizing the reflection coefficient matrix of the RIS becomes highly complicated or even infeasible, particularly when the RIS dimension is very high. 
A more practical approach is to leverage random configurations and optimize the probability distribution of the RIS. 
Additionally, when the user's position is unknown, it is not advisable to focus the beam on a specific point; instead, the scattering capabilities of the RIS should be exploited.
In \Cref{ass:iid}, the RIS is similar to a large scatterer in the environment, leading to fading and phase shift along the path.

The Fisher information matrix in eq.~\eqref{eq:Fisher eq} can be written as
\begin{align}   \label{eq:large RIS}
    \textbf{F} &= \frac{2T }{\sigma^2} \Re \Big\{ (\textbf{D}^H \textbf{D}) \odot \textbf{R}^T \Big\} \notag \\
    &= \frac{2T }{\sigma^2} (\dot{\textbf{A}}_{\Theta}^H ((\textbf{A}_{\Phi}^H \textbf{A}_{\Phi} )\odot \bm{\Sigma})\dot{\textbf{A}}_{\Theta}) \odot \textbf{R}^T,
\end{align}
where $\bm{\Sigma}=\mathbb{E}(\bm{\omega} \bm{\omega}^H)$ represent the auto-correlation matrix of random vector $\bm{\omega}$.
Under the \Cref{ass:iid}, we denote the diagonal elements of $\bm{\Sigma}$ as $E_1 = \mathbb{E}(\bm{\omega}_i \bm{\omega}_i^{*})$
and all off-diagonal elements as $E_2=\mathbb{E}(\bm{\omega}_i) \mathbb{E}( \bm{\omega}_i^{*})$.
Then we have the following proposition.
\begin{proposition}
\label{asymptotic CRB}
For sufficiently large $N$, the $(i,i)$-th element of the Fisher matrix is given as
\begin{align}
\label{asymp Fisher}
    \textbf{F}_{ii} =
\begin{cases}
\frac{p_i}{\sigma^2} \frac{T E_2}{2} N^4& \text{ $ \exists \phi_r, \theta_i+\phi_r=0 $ } \\
\frac{p_i}{\sigma^2} \frac{2T(E_1-E_2)R}{3} N^3& \text{ $ \forall \phi_r,  \theta_i+\phi_r \neq 0 $ }
\end{cases}.
\end{align}
\end{proposition}
Here, $p_i$ represents the signal power of $i$-th target and we denote $\rho_i=\frac{p_i}{\sigma^2}$ as the SNR.

\begin{proof}
\begin{small}
\begin{align}\label{eq:DH D}
   \hspace{-5mm}
    &[\textbf{D}^H\textbf{D}]_{ii} = \tr(\dot{\textbf{a}}(\theta_i)\dot{\textbf{a}}(\theta_i)^H \bm{\Omega}^H \textbf{A}_{\Phi}^H \textbf{A}_{\Phi} \bm{\Omega})),\\
    &= \sum_{r=1}^R \sum_{n_1=0}^{N-1} \sum_{n_2=0}^{N-1} \bm{\omega}_{n_2} \bm{\omega}_{n_1}^{*} n_1 n_2 e^{j(\theta_i+\phi_r)(n_1-n_2)} \notag ,\\
    &= \underbrace{E_2 \sum_{r=1}^R \sum_{n_1=0}^{N-1} \sum_{\substack{n_2=0\\n_2 \neq n_1}}^{N-1} n_1 n_2 \cos(\theta_i+\phi_r)(n_2-n_1)}_{n_1 \neq n_2} +\underbrace{E_1 \sum_{r=1}^R \sum_{n=1}^{N-1} n^2 }_{n_1=n_2} .\notag
\end{align}
\end{small}
Before further proof, we draw the following conclusions.
\begin{remark}
\begin{small}
    \begin{align}\label{re:1}
    \sum_{n_1=1}^{N-1} \sum_{n_2=1}^{N-1} n_1 n_2 \cos((n_2-n_1)\psi) \stackrel{N \to \infty}{\longrightarrow} \frac{1}{4(\sin{\frac{\psi}{2}})^2}\mathcal{O}(N^2)
\end{align}
\end{small}
\end{remark}
\begin{proof}
The proof primarily involves taking the derivatives of $\sum_{n=1}^N \sin(n\psi)$ and $\sum_{n=1}^N \cos(n\psi)$and subsequently substituting them back into the left side of \eqref{re:1}.
\end{proof}

When $\theta_i+\phi_r=0$, denote $\theta_i+\phi_r= \psi$, then
\begin{small}
    \begin{align*} 
\begin{split}
       [\textbf{D}^H\textbf{D}]_{ii} &= \frac{RE_1}{3}\mathcal{O}(N^3)+\frac{E_2}{4}\mathcal{O}(N^4)+\sum_{\substack{r=1\\\theta_i+\phi_r \neq 0}}^R \frac{E_2}{4(\sin{\frac{\psi_r}{2}})^2}\mathcal{O}(N^3)  \\
       &= \frac{E_2}{4}\mathcal{O}(N^4) ,
\end{split}
\end{align*}
\end{small}
where the term $\frac{E_2}{4}\mathcal{O}(N^4)$ arises from the condition $\theta_i+\phi_r=0$.

When $\theta_i+\phi_r \neq 0$,  then
\begin{small}
\begin{align*} 
\begin{split}
       [\textbf{D}^H\textbf{D}]_{ii}&= \frac{RE_1}{3}\mathcal{O}(N^3)+\sum_{\substack{r=1\\\theta_i+\phi_r \neq 0}}^R \frac{E_2}{4(\sin{\frac{\psi_r}{2}})^2}\mathcal{O}(N^2) -\frac{RE_2}{3}\mathcal{O}(N^3) \\
       &= \frac{(E_1-E_2)R}{3}\mathcal{O}(N^3),
\end{split}
\end{align*} 
\end{small}

Proposition \ref{asymptotic CRB} is then proved by substituting the aforementioned equations into the formulation of eq. \eqref{eq:large RIS}.
\end{proof}

\begin{remark}
Proposition \ref{asymptotic CRB} provides an analytical formula detailing the relationship between the Fisher matrix and parameters, including the RIS dimension, the statistical distribution of RIS phases, the number of sensors, SNR, etc. 
Notably, when the signals are uncorrelated, we obtain a simplified expression for the CRB.
\begin{align} \label{eq:asypCRB}
    \text{CRB}_{i,i} =
\begin{cases}
\frac{2}{\rho_i T E_2} N^{-4}& \text{ $ \exists \phi_r, \theta_i+\phi_r=0 $ } \\
\frac{3}{2T\rho_i(E_1-E_2)R}N^{-3}& \text{ $ \forall \phi_r,  \theta_i+\phi_r \neq 0 $ }
\end{cases}.
\end{align}

It demonstrates that when the positions of the target and the sensor are symmetrical, i.e. $\theta_i+\phi_r=0$, the CRB decreases with the fourth power of the dimension of the RIS.
In cases where $\theta_i+\phi_r \neq 0$ for all $\theta$ and $\phi$, the CRB decreases in the third-order of RIS dimension.
One possible reason is that when the transmitter and receiver are symmetrical, the signal strength at the receiver is strongest under the random RIS setting, which contributes to improved perception performance.
An interesting observation is that the CRB is independent of the specific values of $\bm{\theta}$ and $\bm{\phi} $.
\end{remark}

\begin{proposition}\label{order relationship}
The CRB matrix satisfies the following order relationship:
\begin{align*}
    \text{CRLB}(T) & \succeq \text{CRLB}(T+1)  \\
    \text{CRLB}(N) & \succeq \text{CRLB}(N+1)
\end{align*}
\end{proposition}
\begin{proof}
The proof follows a methodology akin to Appendix F in \cite{stoica1989music}.
\end{proof}
\Cref{order relationship} shows that CRB decreases with the increase of RIS dimension and the number of time slots.

In conclusion, in this section, we have obtained the scaling law of CRB and order relations with the system parameters.
The asymptotic CRB formulas are much easier to evaluate than the exact CRB expression in finite cases, and provide a good approximation to the exact CRB in large RIS systems.

\subsection{Large system and large RIS}
In this section, we explore the asymptotic properties of CRB when the number of RIS elements, the number of targets, and the sensors are significantly large. 
We begin by setting the following assumptions.
\begin{assumption}\label{ass:widelydoa}
The DoAs $\theta_1,\cdots,\theta_K, \phi_1,\cdots,\phi_R$ are fixed as $N \to \infty$ and widely spaced, which corresponds to the case where the DoAs have angular separation much larger than a beam-width $2\pi/N$.
\end{assumption}
\begin{assumption}\label{ass:large array}
As $R, K, N \to \infty$, we have $K/N \to c_1 \in (0,\infty)$ and $R/N \to c_2 \in (0,\infty)$.
\end{assumption}
\begin{assumption}\label{ass:uncorrelated}
Correlated signals with equal power: the signal covariance matrix $\mathbf{R}$ is an all-one matrix.
\end{assumption}

In this scenario, we relax the constraints imposed by \Cref{ass:iid}.
The array elements of RIS may follow a certain joint distribution or be independent.
However, due to the interdependence introduced by the expectation in \eqref{eq:DH D}, the validity of \eqref{asymp Fisher} no longer holds true. 
The complexity of this model is attributed to the intricate structure arising from the multiplication of high-dimensional matrices.
An efficient technology is the random matrix theory (RMT). 

We draw inspiration from the asymptotic statistics of the channel capacity of Multiple-Input Multiple-Output (MIMO) systems. 
For a given channel $\textbf{H}$, the channel capacity in the absence of channel state information is:
\begin{align*}
    C = \log \left( \det(\textbf{I}_{N_R}+\frac{\text{SNR}}{N_T}\textbf{H}\textbf{H}^H)  \right),
\end{align*}
where $\frac{\text{SNR}}{N_T}$ represents the SNR per transmitter and $N_T$ represents the number of transmitting antennas.
The channel capacity depends on the eigenvalues $\lambda_k \ge 0$ of the $N_R \times N_R$ matrix $\frac{1}{N_T}\textbf{H}\textbf{H}^H$ as:
\begin{align*}
    C &= \sum_{k=1}^{N_R} \log(1+\text{SNR}\lambda_k),\\
    &= N_R \int_0^{\infty} \log(1+\text{SNR}x)f(x)\, dx,
\end{align*}
where eigenvalue probability density function (pdf) $f(x)=\frac{1}{N_R}\sum_{k=1}^{N_R}\delta(x-\lambda_k)$.
Hence, when the pdf tends almost surely to a deterministic function $f^{\infty}(x)$, also known as the asymptotic eigenvalue probability density function (AEPDF), the asymptotic capacity can be computed in closed form by integrating \cite{749007}.

Inspired by the asymptotic capacity, we propose an analysis framework for asymptotic CRB.
First, we show how to establish the relationship between CRB and the RMT.
The CRB is expressed as follows.
\begin{align} 
    \text{CRB} &= \frac{\sigma^2 }{2TN^4}\tr((\frac{1}{N^4}\textbf{D}^H \textbf{D})^{-1})   ,\\
    &= \frac{\sigma^2 }{2TN^4} \int \frac{1}{\lambda} d \mu_{C}(\lambda),  \label{eq:crb}
\end{align}
where $\lambda_i$ are the eigenvalues of matrix $\frac{1}{N^4} \textbf{D}^H \textbf{D}$, and $\mu_C(\lambda)$ is the AEPDF of $\frac{1}{N^4} \textbf{D}^H \textbf{D}$ defined as $\mu_C(\lambda)=\frac{1}{K}\sum_{i=1}^{K}\delta(\lambda-\lambda_i)$, also referred to as the limiting eigenvalue distribution.
It can be observed that the value of CRB depends on the integral of the AEPDF.

A key tool in RMT to analyze the asymptotic eigenvalue pdf is the Stieltjes transform, namely, $m(z)=\int \frac{1}{\lambda-z}d\mu(\lambda)$, and its inverse transform, denoted as $\mu(z)=\frac{1}{\pi}\lim_{y \to 0^{+}} \mathcal{I}[m(x+jy)]$, where $ \mathcal{I}(\cdot)$ represents the imaginary part of the random variable involved.
However, it is challenging to directly derive the Stieltjes transform of $\frac{1}{N^4} \textbf{D}^H \textbf{D}$ due to the matrix multiplication structure.
Free probability theory states that matrices $\textbf{A}_{\Phi},\bm{\Omega}$ and $ \dot{\textbf{A}}_{\Theta}$ are asymptotically free \cite{tulino2004random}, and thus the AEPDF of $\frac{1}{N^4} \textbf{D}^H \textbf{D}$ is the free multiplicative convolution \cite{6966004} of the AEPDF of three Hermittian matrices: $ \frac{1}{N^3}\dot{\textbf{A}}_{\Theta}^H\dot{\textbf{A}}_{\Theta}$, $\frac{1}{N}\textbf{A}_{\Phi} \textbf{A}_{\Phi}^H$ and $\bm{\Omega} \bm{\Omega}^H$.

We first consider the $\mathcal{N}$-transform of $\mu$ given by \( N_{\mu}(z)=\psi_{\mu}^{-1}(z) \), where $\psi_{\mu}(z)$ is closely related to the Stieltjes transform by $\psi_{\mu}(z)=\int \frac{-\lambda}{\lambda-z} d \mu(\lambda) = -1-z m_{\mu}(z)$.
The $\mathcal{N}$-transform has the following two properties \cite{burda2010applying}.
\begin{theorem}
\label{property of N-transform}
    For an $N \times T$ matrix $A$ and a $T \times N$ matrix $B$, assume that $A$ and $B$ are two independent non-commutative random variables and $AB$ is Hermitian than the N-transform of $AB$ obeys the following ‘non-commutative multiplication law’ and 'cyclic property':
    \begin{align*}
        \mathcal{N}_{AB}(z) &= \frac{z}{1+z} \mathcal{N}_{A}(z) \mathcal{N}_{B}(z) , \\
        \mathcal{N}_{AB}(z) &= \mathcal{N}_{BA}(\frac{N}{T}z).
    \end{align*}
\end{theorem}

For asymptotically free matrices, the limiting eigenvalue distribution of matrix product satisfies the free multiplicative convolution denoted $\boxtimes$, namely,
\begin{align*}
    \mu_{\textbf{D}^H \textbf{D}} = \mu_{\frac{1}{N^3}\dot{\textbf{A}}_{\Theta}^H\dot{\textbf{A}}_{\Theta}} \boxtimes \mu_{\bm{\Omega} \bm{\Omega}^H} \boxtimes \mu_{\frac{1}{N}\textbf{A}_{\Phi} \textbf{A}_{\Phi}^H}.
\end{align*}
The $\mathcal{N}$-transform of $\frac{1}{N^4}\textbf{D}^H \textbf{D}$ can be calculated using Theorem \ref{property of N-transform} as follows:
\begin{align} 
    &\mathcal{N}_{\frac{1}{N^4}\textbf{D}^H \textbf{D}}(z) = \mathcal{N}_{\frac{1}{N^4}\dot{\textbf{A}}_{\Theta} \dot{\textbf{A}}_{\Theta}^H \bm{\Omega}^H \textbf{A}_{\Phi}^H \textbf{A}_{\Phi} \bm{\Omega}}(c_1 z),  \notag  \\  
    &= \frac{c_1 z}{1+c_1 z} \mathcal{N}_{\frac{1}{N^3}\dot{\textbf{A}}_{\Theta} \dot{\textbf{A}}_{\Theta}^H}(c_1 z) \mathcal{N}_{\frac{1}{N}\bm{\Omega} \bm{\Omega}^H \textbf{A}_{\Phi}^H \textbf{A}_{\Phi}}(c_1 z)  \label{eq:N-transform},\\
    &= (\frac{c_1 z}{1+c_1 z})^2 \mathcal{N}_{\frac{1}{N^3}\dot{\textbf{A}}_{\Theta}^H \dot{\textbf{A}}_{\Theta}}(z) \mathcal{N}_{\bm{\Omega} \bm{\Omega}^H}(c_1 z) \mathcal{N}_{\frac{1}{N}\textbf{A}_{\Phi} \textbf{A}_{\Phi}^H}(\frac{c_1}{c_2}z) . \notag
\end{align}

Then $\psi(z)$ of $\frac{1}{N^4}\textbf{D}^H \textbf{D}$ can be derived by solving the following integral equation:
\begin{align}
    \begin{split}
        c_1 \psi(z) = \int \frac{\lambda \quad d \mu_{\bm{\Omega} \bm{\Omega}^H}(\lambda)}{\frac{z}{\left(\frac{c_1\psi(z)}{1+c_1\psi(z)}\right)^2 \mathcal{N}_{\frac{1}{N^3}\dot{\textbf{A}}_{\Theta}^H \dot{\textbf{A}}_{\Theta}}\left(\psi (z)\right) \mathcal{N}_{\frac{1}{N}\textbf{A}_{\Phi} \textbf{A}_{\Phi}^H}\left(\frac{c_1}{c_2}\psi (z)\right)} - \lambda},\label{eq:psi}
    \end{split}
\end{align}
and the corresponding $m(z)$ by substituting $(-1-\psi(z))/z$.
Finally, we can determine the limiting eigenvalue distribution $\mu^{\infty}(\lambda)$ of $\frac{1}{N^4}\textbf{D}^H \textbf{D}$ through the inverse Stieltjes transform.
Then the asymptotic CRB can be obtained by substituting in \eqref{eq:crb} $\mu^{\infty}(\lambda)$ for $\mu_C(\lambda)$:
\begin{align*} 
    \text{CRB}^{\infty} =\frac{\sigma^2 }{2TN^4} \int \frac{1}{\lambda} d \mu^{\infty}(\lambda).
\end{align*}

In the following, we will discuss the impact of the RIS configuration matrix on CRB in different scenarios.


\subsubsection{Constant modulus and random phase model}
\label{case one}
The most common model in current research is the constant modulus model, where the phase response of the RIS is characterized by two common cases: continuous and discrete. 
In the case of the discrete phase, it is assumed that the magnitude response of RIS is 1, and the phase is 2bit randomly discretized.
For example, the phase can only take pi/2 or -pi/2 by controlling the ON-OFF status of a PIN diode.
Or we consider the continuous case where the phase of the RIS follows a uniform distribution between $(0,2\pi)$, while maintaining a magnitude response of 1. 
This assumption is often applicable in RIS implementations that utilize varactors or CMOS, allowing for continuous adjustment of the phase.
In both cases, the $\mathcal{N}$-transform of $\bm{\Omega} \bm{\Omega}^H$ is $\mathcal{N}_{\bm{\Omega} \bm{\Omega}^H}(c_1 z) = \frac{1+c_1 z}{c_1 z}$.

Besides, note that under the \Cref{ass:widelydoa}, we have the following asymptotic conclusion:
\begin{align*}
    \textbf{A}_{\Phi} \textbf{A}_{\Phi}^H / N \to \textbf{I}_R, \quad R, N \to \infty \\
    \dot{\textbf{A}}_{\Theta}^H \dot{\textbf{A}}_{\Theta} / N^3 \to \textbf{I}_K / 3, \quad K, N \to \infty,
\end{align*}
which can be proved by making use of the following results \cite{stoica1989maximum}:
\begin{align*}
    \frac{1}{N^{k+1}}\sum_{t=1}^N t^ke^{jt(w_1-w_2)} \stackrel{N \to \infty} {\longrightarrow}
    \begin{cases}
        0&  w_1 \neq w_2 \\
        1/(k+1)&  w_1 = w_2 .
    \end{cases}
\end{align*}
Then we have:
\begin{align*}
    \mathcal{N}_{\frac{1}{N}\textbf{A}_{\Phi} \textbf{A}_{\Phi}^H} (\frac{c_1}{c_2}z) = \frac{c_1 z+c_2}{c_1 z} , \quad \mathcal{N}_{\frac{1}{N^3}\dot{\textbf{A}}_{\Theta}^H \dot{\textbf{A}}_{\Theta}}(z) = \frac{1+z}{3z} .
\end{align*}

Substitute into \eqref{eq:psi} we have the following equation
\begin{align*}
    z = \frac{[1+\psi(z)]\, [c_1 \psi(z)+c_2]}{3\, \psi(z)\, [1+c_1\psi(z)]},
\end{align*}
and the corresponding quadratic equation with respect to $m(z)$ by substituting $\psi_{\mu}(z)= -1-z m_{\mu}(z)$:
\begin{align} \label{eq:quad eq}
    a_2 m^2(z) + a_1 m(z) + a_0 = 0,
\end{align}
where $a_2=-c_1z^2+3c_1z^3, \, a_1 = (c_2-c_1)z+(6c_1-3)z^2,\, a_0=(3c_1-3)z$.
We can find the closed-form expression of $m(z)$ by solving the quadratic equation directly and then $\mu(z)$ can be obtained using the Stieltjes inverse transform.


\subsubsection{Discrete amplitude model}
Another commonly used model is the discrete amplitude model, where the amplitude is constrained to binary values $x$ and $y$ with a probability of $p$, that is to say:
\begin{align*}
    P(\beta_i = x ) = p, \quad P(\beta_i = y ) = 1- p.
\end{align*}
At the same time, the phase response obeys a uniform distribution between $(0,2\pi)$.
This assumption applies to the novel concept of simultaneously transmitting and reflecting RISs (STAR-RISs)\cite{mu2021simultaneously}, where only a subset of elements are selected for transmission and reflection. 
Therefore, when the sensor and target are on the same side of the RIS, only the elements operating in reflection mode play a major role in the incident signal.
Therefore, the coefficient $p$ also represents the proportion of reflective elements.
The eigenvalue pdf of the matrix $\bm{\Omega} \bm{\Omega}^H$ is given by:
\begin{align*}
    \mu_{\bm{\Omega} \bm{\Omega}^H}=\frac{Np\delta(\lambda-x^2)+N(1-p)\delta(\lambda-y^2)}{N},
\end{align*}
where $\delta(\cdot)$ represents the impulse function.
Substituting the above equation into \eqref{eq:psi}, and utilizing the relationship $\psi_{\mu}(z)= -1-z m_{\mu}(z)$, we can get the following 5th-degree polynomial equation in terms of $m(z)$:
\begin{align*}
     c_1+c_1zm(z) =& \frac{px^2}{\frac{(1-c_1-c_1zm(z))^2}{c_1m(z)(c_2-c_1-c_1zm(z))}-x^2} \\
     & + \frac{(1-p)y^2}{\frac{(1-c_1-c_1 z m(z))^2}{c_1m(z)(c_2-c_1-c_1zm(z))}-y^2} .
\end{align*}
Similarly, $\mu(z)$ can be solved numerically using the Stieltjes inverse transform.

In summary, in this section, we analyze the asymptotic CRB, which is contingent on the limiting eigenvalue distribution of the Fisher matrix. 
When the number of sensors and targets both tends to infinity, the CRB becomes independent of the specific values of DoAs but relies on the distribution of the RIS. 
For both constant and discrete amplitude models with random phase between $(0,2\pi)$, we demonstrate that the CRB converges to an asymptotic CRB that depends on the RIS dimension, the number of sensors, and the number of targets.

\section{numerical results}
\subsection{The impact of angles of receivers}
We consider five receivers to estimate three targets, and the response matrix of RIS changes in a uniform distribution manner, that is, $K=3$, $R=5$, $\bm{\Omega}_{ii}\sim U(0,\pi)$.
We consider the uncorrelated signal case and the correlation matrix $\textbf{R}$ is equal to $\textit{\textbf{I}}_K$.

We set the angle of $\theta_1$ to satisfy the condition that there exists $\phi_r$ such that $\theta_1+\phi_r=0$, and we set the angle of $\theta_2$ such that $\theta_2+\phi_r \neq 0$ for all $\phi_r$.
The values of theta and phi are randomly generated according to the specified requirements.
The theoretical values of CRB corresponding to DoA $\theta_1$ and $\theta_2$ are marked by red and blue circles respectively, in Fig.~\ref{fig1}.
The corresponding numerical results of CRB are represented by red lines and blue lines, respectively.

The results show that CRB decreases in the fourth-order of $N$ when there exists $\phi_r$ such that $\theta_1+\phi_r=0$, and decreases in the third-order of $N$ when $\theta_2+\phi_r \neq 0$ for all $\phi_r$, which is consistent with the theoretical expression in \eqref{eq:asypCRB}. 
An interesting finding is that mirroring the angles helps to improve localization performance.
\begin{figure}[h]  
  \centering
  \includegraphics[width=0.8\linewidth]{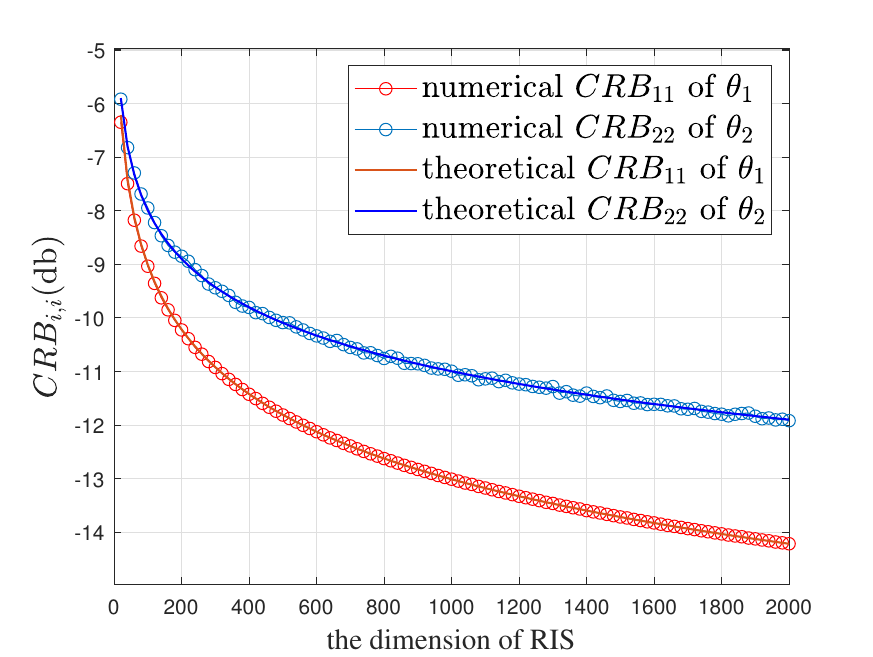}
    \caption{ The value of CRB for DoAs $\theta_1$ and $\theta_2$ versus the dimension of RIS. $\text{CRB}_{11}$ and $\text{CRB}_{22}$ represent the $(1,1)$ and $(2,2)$-th elements of the CRB matrix. The number of time slots is set as $T=50$ and the SNR is $0$. }  \label{fig1}
\end{figure}

\subsection{The impact of the number of receivers}
In this experiment, we compared the theoretical and numerical CRB under different numbers of sensors $R$.
Take $\text{CRB}_{11}$ as an example, as illustrated in Fig.~\ref{fig2}, it is observed that the CRB decreases with an increase in the number of sensors. 
Furthermore, the asymptotic CRB aligns with the theoretical expression in \eqref{eq:asypCRB}.
\begin{figure}[h]
  \centering
  \includegraphics[width=0.8\linewidth]{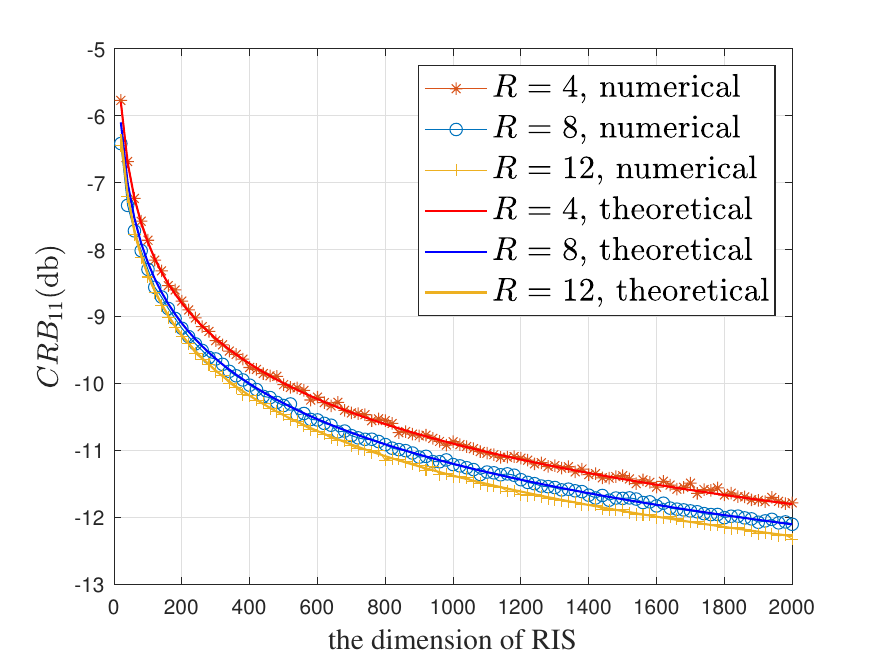}
  \caption{The value of $\text{CRB}_{11}$ versus RIS dimension under different $R$ settings. The parameters remain consistent with those specified in Fig.~\ref{fig1}.}\label{fig2}
\end{figure}

\subsection{Theoretical eigenvalue distribution}
\begin{itemize}
\item \textit{Constant modulus and random phase model:}
we consider the constant modulus model and assume that the phase follows a uniform distribution from $0$ to $2\pi$. 

\item \textit{Discrete amplitude model:}
the amplitude is constrained to binary values $1$ and $3$ with a probability of $0.5$, and the phase follows a uniform distribution from $0$ to $2\pi$. 
\end{itemize}

The theoretical eigenvalue probability density function of $\frac{1}{N^4}\textbf{D}^H \textbf{D}$ is marked by the red line, and the numerical values are marked by the blue histogram. 
The x-axis $\lambda$ represents the eigenvalues, and the y-axis $\mu(\lambda)$ represents the eigenvalues distribution.
It can be seen from Fig.~\ref{fig1:subfig2} that the theoretical eigenvalue PDF is consistent with the numerical value.
The proposed method can accurately characterize the eigenvalue distribution of the Fisher matrix.

\begin{figure}[ht]
  \centering
  \begin{minipage}[t]{0.8\linewidth}  
      \centering
      \label{fig:subfig5}\includegraphics[width=1\textwidth]{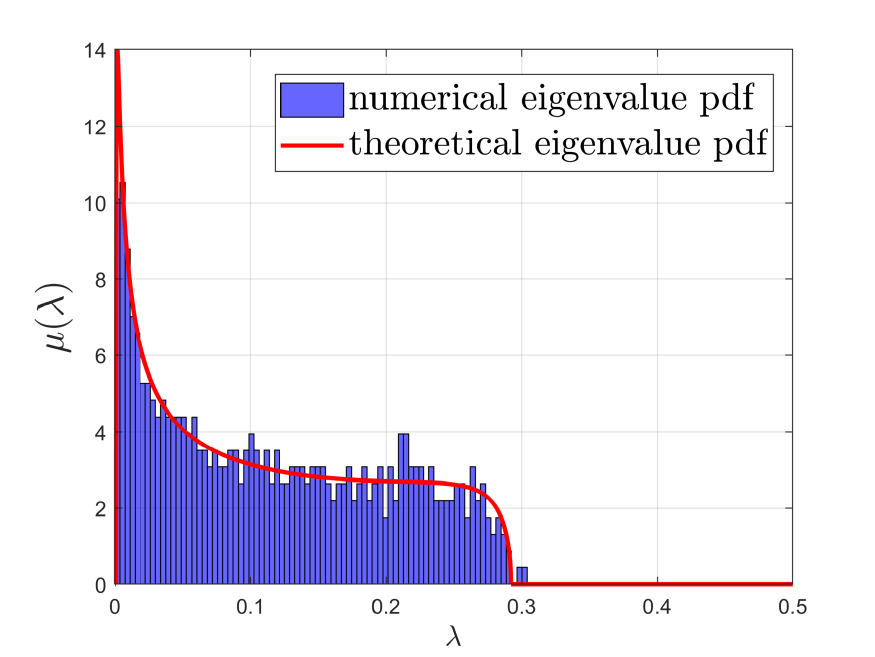}
  \end{minipage} \begin{minipage}[t]{0.8\linewidth}
      \centering
      \label{fig:subfig6}\includegraphics[width=1\textwidth]{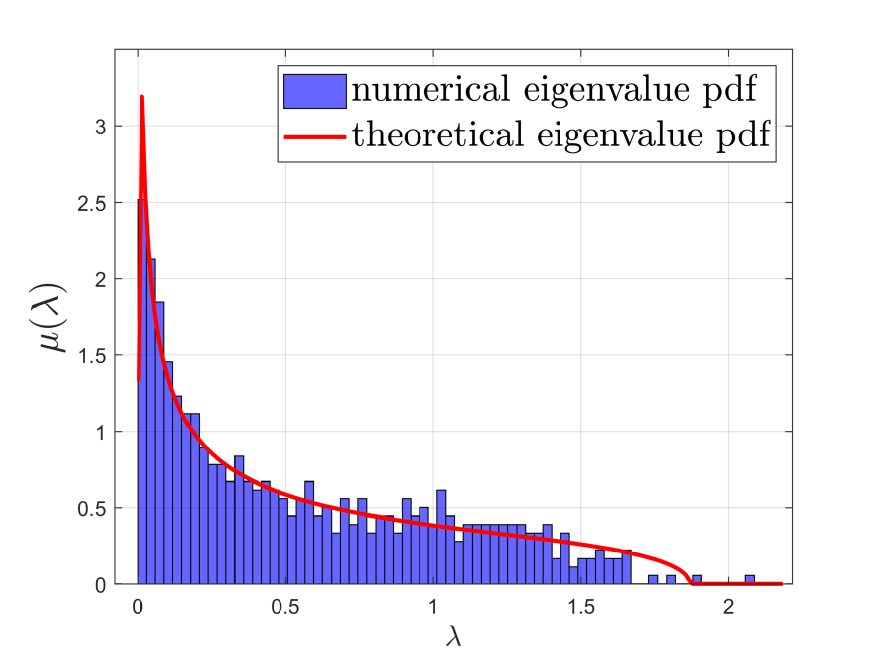}  
  \end{minipage}
  \caption{The theoretical eigenvalue distribution versus the numerical eigenvalue distribution for (\textbf{top}) constant modulus model and (\textbf{bottom}) discrete amplitude model, with $N = 2000$, $K=600$ and $R=700$.}
  \label{fig1:subfig2}
\end{figure}

\section{Conclusion}
This article considers a randomly RIS-assisted multi-user DoA estimation system. 
The first scenario is a large-scale RIS system, where the number of RIS elements tends to be infinite. 
We derive the expression for the asymptotic CRB, revealing the relationship between the asymptotic CRB and system parameters such as RIS dimension, number of sensors, SNR, etc., along with the corresponding scaling law.
In the second scenario, we consider a large system with an abundance of sensors and targets. 
Here, we introduce a CRB analysis framework based on RMT and further analyze the CRB results in the case of constant and discrete amplitude with random phase distribution.
The obtained asymptotic CRB expression provides insight for the performance analysis of random RIS-assisted DoA estimation. 
Finally, simulation results validate the accuracy of the asymptotic CRB expressions.

\vspace{12pt}
\bibliographystyle{IEEEtran}
\bibliography{ref}

\end{document}